\documentclass[amsmath,amsthm,amssymb,article,aps,pra,twocolumn,showpacs,longbibliography]{revtex4-1}

\usepackage[latin1]{inputenc}
\usepackage{graphicx}% Include figure files
\usepackage{dcolumn}% Align table columns on decimal point
\usepackage{bm}% bold math
\usepackage{color}
\newcommand{\ket}[1]{\left|#1\right\rangle}
\newcommand{\bra}[1]{\left\langle#1\right|}

\newtheorem{theorem}{Theorem}
\newtheorem{lemma}{Lemma}

\newenvironment{proof}[1][Proof]{\begin{trivlist}
\item[\hskip \labelsep {\bfseries #1}]}{\end{trivlist}}

\def \qed{$\blacksquare$}

\DeclareMathOperator{\Per}{Per}

\begin{document}

\title{Linear optics only allows every possible quantum operation \\ for one photon or one port}
\author{Julio Jos\'e Moyano-Fern\'andez}
\email{moyano@uji.es}
\affiliation{Universitat Jaume I, Departamento de Matem\'aticas and IMAC-Institut Universitari de Matem\`atiques i Aplicacions de Castell\'o, 12071 Castell\'on de la Plana, Spain.}
\author{Juan Carlos Garcia-Escartin}
\email{juagar@tel.uva.es}
\affiliation{Dpto. de Teor\'ia de la Se\~{n}al y Comunicaciones e Ingenier\'ia Telem\'atica. ETSI de Telecomunicaci\'on. Universidad de Valladolid. Campus Miguel Delibes. Paseo Bel\'en 15. 47011 Valladolid. Spain.}
\date{\today}
\begin{abstract}
We study the evolution of the quantum state of $n$ photons in $m$ different modes when they go through a lossless linear optical system. We show that there are quantum evolution operators $U$ that cannot be built with linear optics alone unless the number of photons or the number of modes is equal to one. The evolution for single photons can be controlled with the known realization of any unitary proved by Reck, Zeilinger, Bernstein and Bertani. The evolution for a single mode corresponds to the trivial evolution in a phase shifter. We analyze these two cases and prove that any other combination of the number of photons and modes produces a Hilbert state too large for the linear optics system to give any desired evolution.
\end{abstract}

\pacs{ 42.50.-p, 42.79.-e, 02.10.Ox}
\maketitle
\section{Quantum optics in photon-preserving linear systems}
\label{intro}
There are many optical elements that can affect the quantum state of light. Elements that preserve the number of photons are particularly interesting in quantum optics and in applications to quantum information \cite{CST89,Leo03,KMN07}. Linear, lossless, passive systems have received a great deal of attention since the demonstration that, combined with measurement, they can be used to build a universal quantum computer \cite{KLM01}. Recently there has been a revived interest kindled by the result that the output statistics of linear optics multiports cannot be accurately predicted in a classical computer efficiently unless several well-founded computational complexity hypothesis are false \cite{AA11}.

In this paper, we study the behaviour of optical systems that act on $n$ photons in $m$ different modes. We call $m\times m$ multiports to the optical systems of interest. The evolution of the state of the photons can be characterized from the scattering matrices $S$ used to describe $m$-ports in classical electromagnetism. We stick to the port denomination for the intuitive picture it gives, but the photons can really be in different orthogonal modes. The key is that two photons in different modes are perfectly distinguishable and do not interfere. The simplest example is a system with photons travelling in different paths, but we can also imagine photons in orthogonal polarization states or which have orthogonal orbital angular momentum states. 

The inputs to our system are a combination of states with $n_i$ photons in a mode with index $i$, denoted by $\ket{n_i}_{i}$. For a system with a total number of photons $n$, all the possible input states can be described as a linear combination of states
\begin{equation}
\label{basisstates}
\ket{\Psi}=\ket{n_1}_{1}\ket{n_2}_{2}\ldots\ket{n_m}_{m}
\end{equation}
with $n_1+n_2+\ldots+n_m=n$. Linear optics multiports present at their output a linear combination of states of the same form. 

The evolution of a photonic quantum state in our system can be specified from a unitary matrix $U$ so that $\ket{\Psi_{\text{out}}}=U\ket{\Psi_{\text{in}}}$. The classical scattering matrix $S$ is enough to characterize the evolution of any number of photons entering the multiport. Both $S$ and $U$ must be unitary matrices as they describe systems that conserve energy and the total probability, respectively. 

The step from $S$ to $U$ depends on the number of photons. If we take the basis composed of the number states of Eq.~($\ref{basisstates}$), the element of $U$ that describes the transition from $\ket{\Psi_{\text{in}}}=\ket{n_1}_{1}\ket{n_2}_{2}\ldots\ket{n_m}_{m}$ to $\ket{\Psi_{\text{out}}}=\ket{n_1'}_{1}\ket{n_2'}_{2}\ldots\ket{n_m'}_{m}$ can be determined from $\bra{n_1'}_{1}\bra{n_2'}_{2}\ldots\bra{n_m'}_{m}U\ket{n_1}_{1}\ket{n_2}_{2}\ldots\ket{n_m}_{m}$, which has a value
\begin{equation}
\label{pereq}
\frac{\Per(S_{\text{in},\text{out}})}{\sqrt{n_1!'\cdot n_2'! \cdots n_m'!\cdot n_1!\cdot n_2! \cdots n_m!}}.
\end{equation} 
In Eq.~($\ref{pereq}$), $\Per(S_{\text{in},\text{out}})$ is the permanent of a matrix $S_{\text{in},\text{out}}$ with elements $S_{i,j}$ from $S$ such that each row index $i$ appears exactly $n_i'$ times and each column index $j$ is repeated exactly $n_j$ times \cite{Sch04,AA11}.

Alternatively, we can write our number states from their creation operators so that $\ket{n_i}_i=\frac{\hat{a}_i^{\dag}}{\sqrt{n_i!}}\ket{0}_i$ and see how the operators transform. For a linear optics multiport, we know \cite{SGL04} the creation operator $\hat{a}_i^{\dag}$ evolves into
\begin{equation}
\label{operatorev}
\sum_{j=1}^{m} S_{ji}\hat{a}_j^{\dag}.
\end{equation}
 
The size of the scattering matrix is a function of the number of inputs and outputs of the optical system. $S$ is an $m\times m$ matrix, whereas $U$ is an $M\times M$ matrix, with $M$ the size of the Hilbert space that contains all the possible configurations of $n$ photons divided into $m$ modes. These different states form a complete basis of the state space and their number is equivalent to the number of ways of placing $n$ indistinct balls in $m$ different boxes, which is the combinatorial number
\begin{equation}
\label{sizeH}
M=\binom{m+n-1}{n}=\frac{(m+n-1)!}{(m-1)!~n!}.
\end{equation}

We can generate all the possible states recursively if we assign a photon number $i$ from $0$ to $n$ to the first mode and then generate all the possible states for the $n-i$ remaining photons in the rest of the modes. By the time we arrive to the last mode the assignment is trivial and we can repeat the procedure until we have a complete list.

\section{Universal quantum transformations}
\label{universal}
We say we have universality if, for our number of photons $n$ and the number of modes $m$ of our system, we can generate any desired quantum evolution $U$ in the state space of all the possible distributions of the $n$ photons in the $m$ modes. 

In this paper, we show there are limitations to the quantum transformations $U$ we can create from a linear optics multiport. While we can implement any desired unitary scattering matrix $S$ using only beam splitters and phase shifters \cite{RZBB94}, a tailored $S$ can only produce any arbitrary $U$ in a limited set of cases.  
 
This is a problem different from finding a universal set of gates for quantum computation. In most linear optics implementations of quantum computing we restrict ourselves to only a subset of all the possible quantum states and there is some kind of postselection. 

\section{Degrees of freedom and universality}
\label{main}
The main result of the paper is a proof that there is a necessary condition for universality which is only satisfied in a limited number of cases for which there are explicit ways to describe how we can generate any desired $U$.

The basic argument is that the degrees of freedom we have when we build the multiport must be at least equal to the degrees of freedom in the photonic Hilbert space. Otherwise, there will be transformations that are impossible to perform.

\begin{lemma}
\label{mainresult}
A linear optics multiport with $m$ inputs cannot be used to give all the possible quantum evolutions in the state space of $n$ photons in $m$ distinct modes unless $m \geq M$, where $M$ is the dimension of the Hilbert space of the photonic states.
\end{lemma}

\begin{proof}
The unitary group $U(m^2)$ contains the $m\times m$ matrices $S$ that describe the linear optics system and the unitary group $U(M^2)$ contains the $M\times M$ matrices $U$ that describe the quantum evolution of the photons' state. Using the expression of Eq.~($\ref{pereq}$) we can define an homomorphism $\varphi: S \to U $ which maps $U(m^2)$ to $U(M^2)$ \cite{AA11}. We can only reach all the matrices in $U(M^2)$ if $\varphi$ is surjective, which for our unitary groups is equivalent to ask for $\varphi$ to be an epimorphism. The homomorphism can only be surjective if the dimension of the domain of $\varphi$ is at least as large as its codomain. In our problem, the condition is $m^2 \geq M^2$, which, for the ranges we are interested in, reduces to the necessary condition for universality

\begin{equation}
\label{neccond}
m\geq\binom{m+n-1}{n}.
\end{equation}

\qed
\end{proof}

The intuition behind this result is that we have only a limited number of degrees of freedom when we build the linear optics system. If the target state space is too big, we cannot reach all the possible matrices $U$. 

In the following sections, we show that in all the cases where necessary condition is met ($n=0$, $n=1$ and $m=1$), there is also an explicit way to find any desired unitary. For $n> 1$ and $m > 1$ we prove it is impossible to implement all possible unitary matrices $U$ using linear optics alone.

\subsection{The vacuum state is always taken to the vacuum}
\label{vacuum}
The first trivial result is that linear optics preserves the vacuum state with zero photons. This is obvious as a passive linear optics multiport cannot create photons, but can also be deduced from the necessary condition of Eq.~($\ref{neccond}$). Our Hilbert space has a dimension 
\begin{equation}
M=\binom{m+n-1}{n}=\binom{m-1}{0}=1
\end{equation}
and $m\geq 1$ for any linear optics system, which will have, at least, one input. There can be many unused degrees of freedom. With no photons the exact configuration of the linear optics multiport is irrelevant and we can choose different scattering matrices. 

\subsection{Systems with one port are equivalent to a phase shifter and trivially give universality for any number of photons}
\label{phaseshifter}
When $m=1$ we have a similar situation. For any number of photons $n$ 
\begin{equation}
M=\binom{m+n-1}{n}=\binom{n}{n}=1
\end{equation}
and the necessary condition of Eq.~($\ref{neccond}$) is met with $m=1=M$. The interpretation is also clear. If we have only one mode, the only allowed physical operation is a phase shift which is equivalent to a $1\times1$ unitary matrix $S=(e^{i\phi})$ whose only element is a root of unity. The linear optics system can only be a phase shifter. The evolution for $n$ photons is then a phase term $e^{in\phi}$. We can use our degree of freedom $\phi$ to give any output phase shift we want and we have universality.

However, in a quantum state we cannot observe a global phase shift. Phase can only be determined when compared to a reference, like in interference between states. This is similar to the definition of voltage, where only differences of voltage have a physical meaning. The output state $e^{in\phi}\ket{n}_1$ is equivalent to the input state $\ket{n}_1$. There is no measurement that can distinguish between these two states. For a set of measurement operators $\left\{E_o\right\}$, we obtain outcome $o$ with probability $p(o)=\bra{n}_1 E_o^\dag E_o\ket{n}_1$ for $\ket{n}_1$, which is the same result we get for $\bra{n}_1e^{-i\phi}E_o^\dag E_oe^{i\phi}\ket{n}_1$. In this case, any unitary matrix $U$ is, really, equivalent to the identity matrix $I$. 

Notice that the equivalence disappears if we have a reference state. If we had one photon in a reference path, the effect of the phase shifter could be observed with a well designed measurement. But then we would be in a different case with $m>1$. 

\subsection{Linear optics multiports can give any possible quantum transformation for one input photon}
\label{singlephoton}
The next interesting case is the evolution of a single photon in an arbitrary multiport with $m\geq 1$ ports. Here 
\begin{equation}
M=\binom{m+n-1}{n}=\binom{m}{1}=m
\end{equation}
and we fulfill the necessary condition with $m=M$. One basis of the Hilbert space of the photon states is the basis of elements $\ket{i}$ where $i=1,\ldots,m$ is the index of the mode our only photon is in. State $\ket{i}$ corresponds to a column vector filled with zeros and a 1 entry in row $i$. From the definition in Eq.~($\ref{pereq}$) we can see that, for one photon, $\bra{j}U\ket{i}=S_{j,i}$. The permanent $\Per(S_{\text{in},\text{out}})$ is exactly the only element of the matrix $S_{\text{in},\text{out}}=(S_{j,i})$ and $U=S$.

We can then implement any desired unitary directly by choosing the appropriate matrix $S$. There are constructive methods to implement any unitary $S$ using only beam splitters and phase shifters \cite{RZBB94} or only one kind of beam splitter for $m\geq 3$ \cite{BA14,Saw16}.

This result cannot be used to build a scalable universal quantum computer. If we want to implement an algorithm acting on $q$ qubits, we need to use $2^q$ paths to generate all the possible states. This exponential growth prevents a generalized use of linear optics with one photon for quantum computation. 

\subsection{Linear optics alone cannot give any desired quantum transformation for more than one input photon in more than one mode}
\label{impossibility}
Apart from the limited results of the previous sections, in general, linear optics multiports cannot give any desired unitary evolution.

\begin{theorem}
\label{impresult}
A linear optics multiport with $m>1$ inputs cannot be used to give all the possible quantum evolutions in the state space of $n>1$ photons in $m$ distinct modes.
\end{theorem}

\begin{proof}
We consider all the cases with $m>1$ and $n>1$. From Eq.~($\ref{sizeH}$)
\begin{align}
M=\binom{m+n-1}{n}=\frac{(m+n-1)\cdots (m+1)\cdot m }{n!}\nonumber\\=m\frac{(m+1)\cdots (m+n-1)}{2\cdot3\cdots n}.
\end{align}
We can write the dimension of the photons' Hilbert space as
\begin{equation}
M\!=m\!\cdot\!\frac{m+1}{2}\cdot\frac{m}{3}\cdots\frac{m+n-1}{n}=m\prod_{k=2}^{n}\left(\frac{m-1+k}{k}\right),
\end{equation}
with a product of terms $1+\frac{m-1}{k}>1$ if $m>1$. For $n>1$ and $m>1$ there is at least one such term in the product and it is immediate to prove $M>m$ which violates the necessary condition of Eq.~($\ref{neccond}$).
\qed
\end{proof}

\section{Comments and examples}
We have shown that, except for a few restricted cases, linear optical systems cannot be used to give any desired quantum evolution for $n$ photons divided into $m$ optical modes. We have given a necessary condition for universality and proved that when the condition is satisfied there are explicit constructions for any unitary evolution $U$ we require. 

It is still open how severe this restriction is. In the condition $m\geq M$, the growth of $M$ as $\binom{m+n-1}{n}$ suggests a smaller number of achievable operators $U$ for higher values of $n$ and $m$. However, there can be important limitations even for small state spaces. We can show that in a simple example with two input ports and two photons. 

The linear optics system is determined by the unitary matrix
\begin{equation}
S=\left( \begin{array}{cc}
S_{11} & S_{12}  \\
S_{21} & S_{22} \end{array} \right).
\end{equation}
From Eq.~($\ref{operatorev}$), we can see an input state $\ket{n_1}_1\ket{n_2}_2$ has an output
\begin{equation}
\frac{1}{\sqrt{n_1!n_2!}}(S_{11} \hat{a}_1^{\dag}+S_{21} \hat{a}_2^{\dag} )^{n_1}(S_{12} \hat{a}_1^{\dag}+ S_{22} \hat{a}_2^{\dag})^{n_2} \ket{0}_1\ket{0}_2.
\end{equation}
For the basis $\left\{\ket{2}_1\ket{0}_2,\ket{0}_1\ket{2}_2,\ket{1}_1\ket{1}_2 \right\}$ and defining
\begin{multline}
\ket{2}_1\ket{0}_2=\ket{20}=\left( \begin{array}{c}
1\\
0\\
0 \end{array} \right), \hspace{3ex} 
\ket{0}_1\ket{2}_2=\ket{02}=\left( \begin{array}{c}
0\\
1\\
0 \end{array} \right),\nonumber\\
\ket{1}_1\ket{1}_2=\ket{11}=\left( \begin{array}{c}
0\\
0\\
1 \end{array} \right),\hspace{10ex}
\end{multline}
the unitary matrix that gives the evolution of the photons' state is
\begin{equation}
\label{U22}
U=\left( \begin{array}{ccc}
S_{11}^2 & S_{12}^2 & \sqrt{2}S_{11}S_{12} \\
S_{21}^2 & S_{22}^2 & \sqrt{2}S_{21}S_{22} \\
\sqrt{2}S_{11}S_{21} & \sqrt{2}S_{12}S_{22} & S_{11}S_{22}+S_{12}S_{21} \end{array} \right).
\end{equation}

We can review many interesting known phenomena from this description. Take for instance the Hadamard matrix that corresponds to a balanced beam splitter
\begin{equation}
S=\frac{1}{\sqrt{2}}\left( \begin{array}{cc}
1 & 1  \\
1 & -1 \end{array} \right).
\end{equation}
For an input state $\ket{11}$, if we substitute the relevant terms in Eq.~($\ref{U22}$) and operate, we get the evolution
\begin{equation}
U\ket{11}=\frac{\ket{20}-\ket{02}}{\sqrt{2}}.
\end{equation}
This is the simplest example of quantum interference between indistinguishable photons and it is described in the famous Hong-Ou-Mandel experiment \cite{HOM87}.

We can also use this simple example to show there are forbidden operations. For instance, the evolution 
\begin{equation}
\label{antidig}
U=\left( \begin{array}{ccc}
0 & 0 & 1 \\
0 & 1 & 0 \\
1 & 0 & 0 \end{array} \right)
\end{equation}
is impossible as we cannot, among others, make $S_{21}^2=0$ and $\sqrt{2}S_{11}S_{21}=1$ at the same time.

We can go a bit further and give bounds to how close we can get to a given state when we start with a fixed input. If we use the general expression in Eq.~($\ref{U22}$), we can see the output state for an input $\ket{11}$ is
\begin{equation}
\ket{\Phi_{\text{out}}}=\left( \begin{array}{c}
\sqrt{2}S_{11}S_{12} \\
\sqrt{2}S_{21}S_{22} \\
S_{11}S_{22}+S_{12}S_{21} 
\end{array} \right).
\end{equation}
Imagine we want to obtain the output state $\ket{20}$. We know this is impossible because it would require the matrix of Eq.~($\ref{antidig}$), up to a global phase. We can, instead, search for the closest possible state, as measured from the overlap 
\begin{equation}
|\bra{20} U \ket{11}|^2=2|S_{11}|^2 |S_{12}|^2=2 |S_{11}|^2 (1-|S_{11}|^2),
\end{equation} 
where we use $S$ is unitary and therefore $|S_{11}|^2+|S_{12}|^2=1$. We would like to get $|\sqrt{2}S_{11}S_{12}|=1$, but we must settle with maximizing $2 |S_{11}|^2 (1-|S_{11}|^2)$. The entry is maximized for $|S_{11}|^2=\frac{1}{2}$ with a maximum overlap $\frac{1}{2}$, which is exactly the case in the Hong-Ou-Mandel experiment. This example shows the limitations can be severe even for values of $M$ slightly above $m$, like our example with $M=3$ and $n=2$.

We can also wonder if the results are valid outside Fock states. However, other states, like coherent states 
\begin{equation}
\ket{\alpha}=e^{-\frac{|\alpha|^2}{2}}\sum_{k=0}^{\infty}\frac{\alpha^k}{\sqrt{k!}} \ket{k} 
\end{equation}
can always be written as a linear superposition of number states. Linear optics preserves the number of photons and we can study separately the evolution for different photon numbers. For most of the terms in the superposition we cannot achieve any arbitrary evolution. Unless we only have superpositions of states for which universal evolution is possible there will be forbidden operations.

The restrictions of the achievable evolutions $U$ does not mean we cannot produce any desired output state. We can always introduce the desired state $\ket{\Psi}$ in a linear system with a scattering matrix $S$ and measure the output $\ket{\Phi}$. The inverse system, with matrix $S^{\dag}$, will produce an output $\ket{\Psi}$ for an input $\ket{\Phi}$. Trivially, if $S=I$ we can generate any output by choosing that state at the input. This only shows that arbitrary state preparation is equivalent to preparing a known state and being able to perform an arbitrary evolution. 

The interest of the presented result lies in the realization that certain states cannot be achieved from certain inputs. Determining which states can be reached for any given input state is left as an open problem that will require different methods that the ones presented here.

\begin{acknowledgments}
The first author has been partially supported by the Spanish Government Ministerio de Econom\'ia y Competi\-tividad (MINECO), grant MTM2012-36917-C03-03, and by Universitat Jaume I, grant P1-1B2015-02.
\end{acknowledgments}

%\bibliography{librodearena}
\newcommand{\noopsort}[1]{} \newcommand{\printfirst}[2]{#1}
  \newcommand{\singleletter}[1]{#1} \newcommand{\switchargs}[2]{#2#1}

\end{document}